\documentclass[12pt,reqno]{amsart}
\usepackage{amsmath , amssymb }
\usepackage{graphics}

\usepackage[mathscr]{eucal}

\newdimen\unit\newdimen\psep\newcount\nd\newcount\ndx\newbox\dotb\newbox\ptbox
\newdimen\dx\newdimen\dy\newdimen\dxx\newdimen\dyy\newdimen\hgt
\newdimen\xoff\newdimen\yoff
\newcommand\clap[1]{\hbox to 0pt{\hss{#1}\hss}}
\newcommand\vdisk[1]{{\font\dotf=cmr10 scaled #1\dotf.}}
\newcommand\varline[2]{\setbox\dotb\hbox{\vdisk{#1}}\xoff=-.5\wd\dotb
\wd\dotb=0pt\yoff=-.5\ht\dotb\psep=#2\ht\dotb}
\newcommand\varpt[1]{\setbox\ptbox\clap{\vdisk{#1}}\setbox\ptbox
\hbox{\raise-.5\ht\ptbox\box\ptbox}}
\newcommand\cpt{\copy\ptbox}
\newcommand\point[3]{\rlap{\kern#1\unit\raise#2\unit\hbox{#3}}}
\newcommand\setnd[4]{\dx=#3\unit\advance\dx-#1\unit\divide\dx by\psep
\dy=#4\unit\advance\dy-#2\unit\divide\dy by\psep \multiply\dx
by\dx\multiply\dy by\dy\advance\dx\dy\nd=1\advance\dx-1sp
\loop\ifnum\dx>0\advance\dx-\nd sp\advance\nd1\advance\dx-\nd
sp\repeat}

\newcommand\dline[5]{{\nd=#5\hgt=#2\unit\dx=#3\unit\advance\dx-#1\unit
\divide\dx by\nd\dy=#4\unit\advance\dy-#2\unit\divide\dy by\nd
\advance\hgt\yoff\rlap{\kern#1\unit\kern\xoff\loop\ifnum\nd>1\advance\nd-1
\advance\hgt\dy\kern\dx\raise\hgt\copy\dotb\repeat}}}

\newcommand\qellip[4]{{\setnd{0}{0}{#3}{#4}\dx=\unit\dy=0pt\raise\yoff\rlap{%
\kern#1\unit\kern\xoff\raise#2\unit\hbox{\loop\ifnum\dx>0\rlap{\kern#3\dx
\raise#4\dy\copy\dotb}\hgt=\dx\divide\hgt
by\nd\advance\dy\hgt\hgt=\dy \divide\hgt
by\nd\advance\dx-\hgt\repeat\rlap{\raise#4\dy\copy\dotb}}}}}
\newcommand\bez[6]{{\setnd{#1}{#2}{#3}{#4}\ndx=\nd\setnd{#3}{#4}{#5}{#6}
\ifnum\ndx>\nd\nd=\ndx\fi\dx=#3\unit\advance\dx-#1\unit\dy=#4\unit
\advance\dy-#2\unit\dxx=#5\unit\advance\dxx-#1\unit\dyy=#6\unit\advance
\dyy-#2\unit\advance\dxx-2\dx\advance\dyy-2\dy\divide\dxx
by\nd\divide\dyy by\nd\advance\dx.25\dxx\advance\dy.25\dyy\divide\dx
by\nd\divide\dy by\nd \multiply\nd
by2\dx=100\dx\dy=100\dy\dxx=100\dxx\dyy=100\dyy\divide\dxx by\nd
\divide\dyy by\nd\hgt=#2\unit\raise\yoff\rlap{\kern#1\unit\kern\xoff
\raise\hgt\copy\dotb\loop\ifnum\nd>0\advance\nd-1\advance\hgt0.01\dy
\kern0.01\dx\raise\hgt\copy\dotb\advance\dx\dxx\advance\dy\dyy\repeat}}}
\newcommand\ptu[3]{\point{#1}{#2}{\cpt\raise1ex\clap{$\scriptstyle{#3}$}}}
\newcommand\ptd[3]{\point{#1}{#2}{\cpt\raise-1.8ex\clap{$\scriptstyle{#3}$}}}
\newcommand\ptr[3]{\point{#1}{#2}{\cpt\raise-.4ex\rlap{$\ \scriptstyle{#3}$}}}
\newcommand\ptl[3]{\point{#1}{#2}{\cpt\raise-.4ex\llap{$\scriptstyle{#3}\ $}}}
\newcommand\ptlu[3]{\point{#1}{#2}{\raise.8ex\clap{$\scriptstyle{#3}$}}}
\newcommand\ptld[3]{\point{#1}{#2}{\raise-1.6ex\clap{$\scriptstyle{#3}$}}}
\newcommand\ptlr[3]{\point{#1}{#2}{\raise-.4ex\rlap{$\,\scriptstyle{#3}$}}}
\newcommand\ptll[3]{\point{#1}{#2}{\raise-.4ex\llap{$\scriptstyle{#3}\,$}}}

\newcommand\thnline{\varline{400}{.6}}

\varpt{2500}\thnline\unit=2em

\newtheorem{thm}{Theorem}

\newtheorem{conj}{Conjecture}

\newtheorem{lemma}[thm]{Lemma}

\newtheorem{obs}[thm]{Observation}
\theoremstyle{definition}
\theoremstyle{definition}\newtheorem*{defn}{Definition}
\theoremstyle{definition}
\newcommand{\ds}{\displaystyle}
\newcommand{\ul}{\underline}

\def\E{\mathcal{E}}

\def\P{\mathcal{P}}
\def\Q{\mathcal{Q}}

\def\N{\mathbb{N}}

\def\Z{\mathbb{Z}}

\def\Pr{\mathbb{P}}
\def\le{\leqslant}
\def\ge{\geqslant}

\def\eps{\varepsilon}

\begin{document}
\title{Zero-temperature Glauber dynamics on $\Z^d$}

\author{Robert Morris}
\address{Murray Edwards College, The University of Cambridge, Cambridge CB3 0DF, England (Work partly done whilst at the Instituto Nacional de Matem\'atica Pura e Aplicada, Rio de Janeiro, Brazil)} \email{rdm30@cam.ac.uk}
\thanks{The author was supported during this research by MCT grant PCI EV-8C}

\begin{abstract}
We study zero-temperature Glauber dynamics on $\Z^d$, which is a dynamic version of the Ising model of ferromagnetism. Spins are initially chosen according to a Bernoulli distribution with density $p$, and then the states are continuously (and randomly) updated according to the majority rule. This corresponds to the sudden quenching of a ferromagnetic system at high temperature with an external field, to one at zero temperature with no external field. Define $p_c(\Z^d)$ to be the infimum over $p$ such that the system fixates at `$+$' with probability $1$. It is a folklore conjecture that $p_c(\Z^d) = 1/2$ for every $2 \le d \in \N$. We prove that $p_c(\Z^d) \to 1/2$ as $d \to \infty$.
\end{abstract}

\maketitle

\section{Introduction}\label{intro}

Perhaps the most extensively studied model in the statistical physics literature is the Ising model of ferromagnetism on $\Z^d$. Despite this, very little has been proved rigorously about the dynamics of the model, even when the temperature is zero. In particular, it is conjectured that the critical threshold $p_c(\Z^d)$ for fixation at the Gibbs state is equal to $1/2$ in all dimensions, but the best known upper bound, due to Fontes, Schonmann and Sidoravicius~\cite{FSS}, is only $p_c(\Z^d) < 1$. In this article we shall prove that this conjecture holds asymptotically as $d \to \infty$.

We begin with a precise definition of the question being investigated. Let $G$ be a (finite or infinite) graph, and endow each vertex $x \in V(G)$ with a spin $\sigma(x) \in \{+,-\}$, and an independent random exponential clock $C(x)$ (so the probability the clock does not ring in time $[s,s+t]$ is $e^{-t}$). We shall investigate \emph{zero-temperature Glauber dynamics} on $G$, which is the following dynamic process: For each vertex $x \in V(G)$ and each time $t \ge 0$, if the clock $C(x)$ does not ring at time $t$, then the state $\sigma(x)$ remains unchanged; if $C(x)$ \emph{does} ring at time $t$, then $\sigma(x)$ changes to agree with the majority of the neighbours of $x$ in $G$. (If there are an equal number of neighbours in each state, then the new state is chosen uniformly at random.) Our question is the following: Given a probability distribution on the state $(\sigma(x) : x \in V(G)) \in \{+,-\}^{V(G)}$ at time $t = 0$, what happens to the distribution of states as $t \to \infty$? In particular, under what conditions do all vertices end up in the same state?

We shall be interested in the above question when $G = \Z^d$, the $d$-dimensional square lattice, and when the states at time $0$ are chosen according to the Bernoulli distribution. More precisely, let $p \in (0,1)$, and suppose the spins $\sigma(x)$ at time $t = 0$ are chosen independently at random, with $\Pr\big( \sigma(x) \textup{ is `$+$'} \big) = p$ for each $x \in \Z^d$. We say that $\Z^d$ \emph{fixates at} `$+$' if, for each vertex $x \in \Z^d$, there is a time $T(x) \in [0,\infty)$ such that $\sigma(x)$ is `$+$' for all times $t \ge T(x)$. Note that if the system does not fixate then (in general) it is possible to have a mixture of vertices which are eventually `$+$', vertices which are eventually `$-$', and vertices which change state an infinite number of times.

Define
$$p_c(\Z^d) \; := \; \inf\big\{ p \,:\, \Pr\big(\Z^d \textup{ fixates at `$+$'} \big) = 1 \big\}.$$
The case $d = 1$ of this problem was first investigated by Erd\H{o}s and Ney~\cite{EN}, who studied the following, slightly simpler problem. Place a particle on each vertex of $\Z$ except the origin, allow each to perform a (possibly biased, discrete time) random walk on $\Z$, and annihilate any pair of particles which cross paths. They conjectured that, with probability 1, the origin is at some point occupied; in our problem this corresponds to the origin changing state at least once. The conjecture of Erd\H{o}s and Ney was proved by Lootgieter~\cite{Loot} and by Schwartz~\cite{Schw} in discrete and continuous time, respectively. Arratia~\cite{Arr} proved a much stronger result: that, for a wide class of random starting configurations, every site is occupied an infinite number of times. It follows easily from Arratia's theorem that, for any $p \in (0,1)$, in Glauber dynamics on $\Z$ every site changes state an infinite number of times, and hence that $p_c(\Z) = 1$.

For $d \ge 2$ the behaviour of the system is expected to be very different. The following conjecture is folklore.

\begin{conj}[Folklore]\label{folklore}
$$p_c(\Z^d) \; = \; \frac{1}{2}$$
for every $2 \le d \in \N$.
\end{conj}

Although the problem of determining $p_c(G)$ has been studied by many authors, and for various classes of infinite graphs $G$, surprisingly little is known. It is obvious that $p_c(\Z^d) \ge 1/2$, by symmetry, and it is straightforward to show that if $p = 1/2$ then $\Pr\big(\Z^d \textup{ fixates at `$+$'} \big) = 0$, using the fact (from ergodic theory) that fixation at `$+$' has probability either 0 or 1.  Nanda, Newman and Stein~\cite{NNS} proved that moreover, if $p = 1/2$ and $G = \Z^2$, then (almost surely) no vertex fixates, i.e., the state of every vertex changes an infinite number of times. However, even this simple statement is unknown if $d \ge 3$, and on the hexagonal lattice the situation is different, with some vertices fixating at `$+$' and others at `$-$', see~\cite{HN}. 

Glauber dynamics has also been studied in detail on the $d$-regular tree, $T_d$ (see for example~\cite{BKMP,CMart,How,MSW}), but even here very little has been proved about $p_c(T_d)$. Indeed, Howard~\cite{How} showed that $p_c(T_3) > 1/2$, and it was proved by Caputo and Martinelli~\cite{CMart} that $p_c(T_d) \to 1/2$ as $d \to \infty$ (in fact their result is more general, and this statement is straightforward to prove in the zero-temperature case), but for every $d \ge 4$ it is unknown whether or not $p_c(T_d) = 1/2$. For further results and problems about the case $p = 1/2$, on $\Z^d$ and on other graphs, see for example~\cite{CSN,How,NS,SCN,Wu}; for a good account of Glauber dynamics at non-zero temperatures, see~\cite{Mart}.

The best known upper bound on $p_c(\Z^d)$ is due to Fontes, Schonmann and Sidoravicius~\cite{FSS}. They proved, using multi-scale analysis, that $p_c(\Z^d) < 1$, i.e., that for each $d \ge 2$, there is an $\eps = \eps(d) > 0$, such that if $p \ge 1 - \eps$ then fixation at `$+$' occurs with probability 1. They moreover showed that this fixation occurs in time with a stretched exponential tail. The values of $\eps(d)$ they obtain converge rapidly to 0 as $d \to \infty$ (see Theorem~\ref{FSSthm} below), but despite this fact, their result will be a crucial tool in our proof.

We shall prove the following result.

\begin{thm}\label{to1/2}
$$p_c(\Z^d) \; \to \; \frac{1}{2}$$
as $d \to \infty$.
\end{thm}

We remark that the same result also holds (in the limit as $d \to \infty$) if, instead of choosing the state uniformly when the number of `$+$' and `$-$' neighbours are equal, we were to choose it to be `$+$' with probability $\alpha \in (0,1)$ (see also~\cite{FSS}). For simplicity, however, we shall assume throughout that $\alpha = 1/2$, as in the definitions above. We note also that the proof relies on very few properties specific to the lattice $\Z^d$, and so it is likely that the same techniques can be extended to a much wider family of high-dimensional lattices (see Theorem~2.2 of~\cite{Maj}).

We shall moreover give concrete bounds on the rate of convergence of $p_c$. These will be easy to read out from our later results, and are certainly not optimal (since we believe Conjecture~\ref{folklore} to be true). However, for the reader's convenience, we state here the precise result that we shall prove. Let $\eps > 0$ and $d \in \N$, with $\eps^2 d \ge 10^{10}\log d$. Then
$$p_c(\Z^d) \; \le \; \frac{1}{2} \, + \, \eps.$$
We remark that, although the constant $10^{10}$ could be improved somewhat with a little extra effort, the techniques in this paper do not work for small values of $d$.

The proof of Theorem~\ref{to1/2} is based on two couplings of Glauber dynamics on large sub-blocks of $\Z^d$ with bootstrap percolation, a monotone version of Glauber dynamics which has itself been studied extensively (see for example~\cite{AL,BB,BPP,CC,Hol,Sch}), and which we shall define in Section~\ref{sketch}. We shall use powerful tools developed by Balogh, Bollob\'as and Morris~\cite{Maj} (see Lemma~\ref{tool}, below) to show that, after time $O(d^5)$, very few vertices are in state `$-$'. Finally we shall apply the result of Fontes, Schonmann and Sidoravicius~\cite{FSS}. The crucial point, throughout the proof, will be that we shall retain independence except at short distances.

The rest of the paper is organised as follows. In Section~\ref{sketch} we recall the definition of bootstrap percolation and the main results of~\cite{Maj} and~\cite{FSS}, and give a sketch of the proof of Theorem~\ref{to1/2}. In Sections~\ref{firstd} and~\ref{afterd} we prove that by the end of the two couplings (which occurs in time $O(d^5)$), sufficiently many vertices are in state `$+$' that we may apply the method of Fontes, Schonmann and Sidoravicius~\cite{FSS}. Finally, in Section~\ref{proofsec}, we complete the proof of Theorem~\ref{to1/2}.

\section{Bootstrap percolation}\label{sketch}

In this section we describe the main tools we shall use, and give a sketch of the proof of Theorem~\ref{to1/2}. We begin by recalling the result of Fontes, Schonmann and Sidoravicius~\cite{FSS}. The following theorem, which is slightly more general than the one they state, is implicit in their proof (see below). Let $L \in \N$, and partition $\Z^d$ into blocks of size $L^d$ in the obvious way. Let $p \in (0,1)$, and consider the collection $\Omega(L,p)$ of probability distributions on $\{+,-\}^{\Z^d}$ satisfying the following conditions:
\begin{itemize}
\item $\sigma(x) = \sigma(y)$ if $x$ and $y$ are in the same block $B$. (Let $\sigma(B) = \sigma(x)$ for $x \in B$.)\\[-1ex]
\item $\Pr\big(\sigma(B) \textup{ is `$+$'} \big) = p$ for every block $B$.\\[-1ex]
\item Given any collection of blocks $\{B_1,\ldots,B_k\}$ with $\|B_i - B_j\|_\infty \ge 2$ whenever $i \neq j$, the states $\sigma(B_1),\ldots,\sigma(B_k)$ are independent.
\end{itemize}
Now, define
$$p^{(L)}_c(\Z^d) \; := \; \inf \left\{ p \, : \, \Pr\big( \Z^d \textup{ fixates at `$+$'} \big) = 1 \textup{ for every } \sigma \in \Omega(L,p) \right\},$$
where $\sigma$ is the initial distribution of states. Note in particular that $p_c(\Z^d) \le p^{(1)}_c(\Z^d)$.

\begin{thm}[Fontes, Schonmann and Sidoravicius~\cite{FSS}]\label{FSSthm}
There exists an $\eps > 0$ such that, for each $2 \le d \in \N$, and each $L \in \N$,
$$p^{(L)}_c(\Z^d) \; \le \; 1 \, - \, \eps \left( \frac{1}{2L} \right)^{d^2}.$$
\end{thm}

The theorem above follows from a multi-scale analysis, using ideas from $2$-neighbour bootstrap percolation (defined below). Theorem~\ref{FSSthm} is slightly stronger than Theorem~1.1 of~\cite{FSS}, but follows from almost exactly the same proof. Indeed, the definition of $\Omega(L,p)$ above is precisely the `block-dynamics' defined in Section~4 of~\cite{FSS}; the theorem applies to any initial distribution satisfying these conditions. In order to obtain the dependence on $L$ and $d$ in Theorem~\ref{FSSthm}, we adjust the proof in~\cite{FSS} as follows: set $\ell_0 = L$ in (4.1), increase $t_k$ by a factor of $L^d$ in (4.2), and weaken the upper bound (4.8) by a factor of $L$. For inequality (4.6) and Step 1 of the proof we require $q := 1 - p \le \exp( \delta/q^{1/(d-1)} )$, for some polynomial $\delta = \delta(d)$; for Step 2 we require (roughly) that $t_{k+1} \gg (1/q_k) L^d$; and for Step 3 we require $t_k \ll \ell_k$. The first of these inequalities is satisfied if $q \le d^{-O(d)}$, the second and third are satisfied if $q \le L^{-d^2}$. We remark that in fact, by choosing $\ell_k$ much larger, one could improve the bound in Theorem~\ref{FSSthm} to $1 - \eps (Ld)^{-O(d)}$. We shall not need this slight strengthening however; in fact a much weaker bound would suffice.

In order to prove Theorem~\ref{to1/2} we shall replace the first stage of the argument of~\cite{FSS} with a more careful calculation, using ideas from majority bootstrap percolation in high dimensions. We remark that we shall not prove a result corresponding to Theorem~\ref{FSSthm}. Our method uses, and absolutely requires, total independence of initial states.

Before embarking on our sketch, let us recall first some of the tools and ideas of~\cite{Maj}, which will be crucial for the proof. First, given a (finite or infinite) graph $G$ and an integer $r \in \N$, we call \emph{$r$-neighbour bootstrap percolation} on $G$ the following deterministic process. Let $A \subset V(G)$ be a set of initially `infected' vertices, and, at each time step, let new vertices of $G$ be infected if they have at least $r$ infected neighbours, and let infected vertices stay infected forever. Formally, set $A_0 = A$, and
$$A_{t+1} \; := \; A_t \: \cup \: \big\{ v \in V(G) \: : \: |\Gamma(v) \cap A_t| \ge r \big\}$$
for each integer $t \ge 0$. The \emph{closure} of $A \subset V(G)$ is the set $[A] \; = \; \bigcup_t A_t$ of eventually infected vertices. We say that the set $A$ \emph{percolates} if eventually the entire vertex set is infected, i.e., if $[A] = V(G)$. If $G$ is $d$-regular and $r = \lceil d/2 \rceil$, then we call the process \emph{majority bootstrap percolation}.

Bootstrap percolation was introduced by Chalupa, Leath and Reich~\cite{CLR} in 1979, and has since been studied by many authors, most frequently on $\Z^d$ and $[n]^d$, the $d$-dimensional torus on $\{1,\ldots,n\}^d$ (see for example~\cite{AL,Maj,n^d,CC,Sch}), but also on trees~\cite{BPP,BS,FS} and random regular graphs~\cite{BPi,Svante}. The elements of the set $A$ are normally chosen independently at random, and the main problem is to determine the critical threshold, $p_c(G,r)$, at which percolation becomes likely. To be precise, write $P_p(G,r)$ for the probability that $A$ percolates in $r$-neighbour bootstrap percolation on $G$ if the elements of $A$ are chosen independently at random, each with probability $p$, and define
$$p_c(G,r) \; := \; \inf\big\{ p \,:\, P_p(G,r) \ge 1/2 \big\}.$$
Balogh, Bollob\'as and Morris~\cite{Maj} recently proved the following theorem about majority bootstrap percolation on $[n]^d$.

\begin{thm}[Balogh, Bollob\'as and Morris~\cite{Maj}]\label{majthm}
Let $n = n(d)$ be a function satisfying
$$2 \; \le \; n \; = \; 2^{2^{O\left(\sqrt{\frac{d}{\log d}}\right)}},$$
or equivalently, $d \ge \eps(\log \log n)^2\log\log\log n$ for some $\eps > 0$. Then
$$p_c\big( [n]^d, d \big) \; = \; \frac{1}{2} \, + \, o(1)$$ as $d \to \infty$.
\end{thm}

We remark that the lower bound on $d$ guarantees that $[n]^d$ is sufficiently `locally tree-like', in the sense that balls with small radii grow quickly. We shall use this observation again later in the proof of Theorem~\ref{to1/2} (see Lemmas~\ref{F'} and~\ref{inact}). Theorem~\ref{majthm} contrasts with the case where $d$ is fixed, when $p_c([n]^d,d) = o(1)$. For some recent, much more precise results about the case $d$ constant, see~\cite{d=r=3,alldr,CM,Hol,Me2d}.

In order to prove the lower bound in Theorem~\ref{majthm}, the authors introduced the following modified bootstrap process. Let $k,m \ge 0$ and $S^{(0)} \subset V(G)$.
\begin{itemize}
\item If $0 \le j \le k - 1$, then
$$S^{(j+1)} \; = \; S^{(j)} \, \cup \, \big\{ x \, :\, |\Gamma(x) \cap S^{(j)}| \ge r - (k-j)m \big\}.$$
\item If $j \ge k$, then $S^{(j+1)} \, = \, S^{(j)} \, \cup \, \big\{ x \, :\, |\Gamma(x) \cap S^{(j)}| \ge r \big\}.$
\end{itemize}
We call this process $\textup{Boot}(r,k,m)$. Note that it dominates the original process (i.e., the $\textup{Boot}(r,k,0)$ process), in the sense that if the original process percolates, then so does the modified process. It also has the extra property that if the original process does not percolate (and $m$ is chosen correctly), then the modified process \emph{almost always} stops quickly. (For a more precise formulation of this statement, see for example Lemma 6.3 of \cite{Maj}, or Lemma~\ref{tool} below.)

We need one more definition.

\begin{defn}
Given a (possibly infinite) graph $G$, an integer $C \in \N$, and a collection of events $\E = \{E_v : v \in V(G)\}$, one for each vertex of $G$, we say that the events in $\E$ are \emph{$C$-independent} if the following holds. For each $k \in \N$, if $\{v_1,\ldots,v_k\} \subset V(G)$ satisfies $d_G(v_i,v_j) \ge C$ for every $i \neq j$, then the events $\{E_{v_1},\ldots,E_{v_k}\}$ are independent.
\end{defn}

We are now ready to give our sketch of the proof of Theorem~\ref{to1/2}. First let $n = 2^d$, and partition $\Z^d$ into blocks of size $[n]^d$ in the obvious way. Note that $d = \log n \gg (\log \log n)^2\log\log\log n$, so the method of the proof of Theorem~\ref{majthm} will apply to these blocks. Consider the graph $G$ induced by one particular block, $B$. The basic idea is as follows. First we run the majority bootstrap process on $G$, with the infected sites being those initially in state `$-$'. Next we observe that, since (by Theorem~\ref{majthm}) the initial density of `$-$' vertices is subcritical, very `few' vertices change state. Finally, we run Glauber dynamics until all the clocks associated with vertices of $G$ have rung at least once. If the states of the vertices after the bootstrap process were all independent then, by Chernoff's inequality, only about $e^{-\eps^2 d}|B|$ of them would have as many `$-$' neighbours as `$+$' neighbours (since very few have changed state), so almost all should end up in state `$+$'. However, this is not the case: the bootstrap process brings in long-distance dependence between the states. We shall therefore have to be a little more clever.

Indeed, what we actually do is to couple the original process $\P$ up to time $d$, with a process $\Q$, which is \emph{almost always} biased towards state `$-$', but which still finishes with all but (about) $e^{-\eps^2 d}|B|$ vertices in state `$+$', \emph{and} only has short-distance dependencies! The process $\Q$ is as follows. First, run the $\textup{Boot}(d,8,m)$ process for `$-$' vertices in a slightly larger block $B' \supset B$ (in fact $B'$ is larger by a factor of $5/3$), with $m = \eps d/24$, for eight steps only. We remark that the number eight here could be replaced by any $k \ge 8$; we need only that $\eps^{k+2} d^{k+1} \ge d^4$ (see Lemma~\ref{F}).

Now, with probability about $1 - e^{-d^4}$, the set of vertices in state `$-$' thus obtained will be closed under the majority bootstrap process, in which case no other `$+$' vertex in $B$ can ever again change state in $\P$, unless it is affected by vertices outside $B'$, which (we shall show, see Lemma~\ref{F''}) is very unlikely to occur before time $O(d^5)$. We ignore (i.e., assume to be entirely `$-$') those blocks for which either of these bad events holds (i.e., those which are not closed under bootstrap, and those which are affected by the state of some vertex outside $B'$).

Assume from now on that neither of these two bad events holds for the block $B$, and let $X$ be the set of vertices in $B$ which are `infected' during the $\textup{Boot}(d,8,m)$ process. This set contains all of those vertices which are initially in state `$+$', but could potentially change state without being affected by anything outside $B'$. The events $\{x \in X\}_{x \in B}$ are 17-independent, by the definition of the $\textup{Boot}(d,8,m)$ process. Moreover, we shall show, using the method of~\cite{Maj}, that $\Pr(x \in X) \le 2e^{-2\eps^2 d}$ for each $x \in B'$ (see Lemma~\ref{X}).

Now, let a vertex $x \in B'$ be in state `$-$' after the process $\Q$ if either its clock has not yet rung in $\P$ by time $d$, or if it had at least $d$ neighbours in state `$-$' initially, or if it has at least one neighbour in $X$. The probability that at least one of these events occurs is at most
$$e^{-d} + e^{-2\eps^2 d} + 4d e^{-2\eps^2 d} \; \le \; 5de^{-2\eps^2 d} \; < \; \left( \frac{1}{d} \right)^{1000}$$
since $\eps^2 d \ge 10^{10} \log d$ (see Lemma~\ref{qbound}). Moreover, assuming that the two `bad' events defined above do not hold, the set of `$-$' vertices obtained through $\Q$ contains that obtained through $\P$, run up to time $d$ (see Lemma~\ref{F}).

We have shown that up to time $d$, the process $\P$ may be `approximately' coupled with a process in which
$$\Pr\big(\sigma(x) \textup{ is `$-$' after time }d\big) \; \le \; d^{-1000},$$
and the events $\{\sigma(x) \textup{ is `$-$' after time } d\}_{x \in B}$ are 19-independent (we lose a little more independence in going from $X$ to $\Q$). The proof is now completed in three more steps. First, we describe a second coupling, with a process in which the probability a vertex is \emph{ever again} in state `$-$' after time $d$ (unless affected by vertices outside $B'$) is still at most $d^{-500}$, and in which these events are 120-independent (see Lemmas~\ref{F'} and~\ref{z30}). Next we deduce that after time $200d^5 + d$, with very high probability every vertex of $B$ will be in state `$+$' (see Lemma~\ref{inact}). Since $n = 2^d \gg 200d^5$, it is very unlikely that the state of any vertex in $B$ has by this point been affected by any vertex outside $B'$ (see Lemma~\ref{F''}). Finally, we apply Theorem~\ref{FSSthm} to the distribution of states obtained on the blocks $B$.

Throughout the proof we shall have a large amount of leeway in our calculations, and so we shall often be able to use very weak approximations. The crucial point, however, is that the set $X$ must be small (see Lemma~\ref{X}); it is at this step that the proof is sharp.

\section{A coupling up to time $d$}\label{firstd}

In this section we shall prove the required facts about the processes $\P$ and $\Q$. First let us define $\P$ and $\Q$ precisely.

Let $B$ be a block in $\Z^d$ of size $[n]^d$, where $n = 3 \times 2^d$, and let $B'$ be a block with the same centre as $B$, but of size $[n']^d$, where $n' = 5 \times 2^d$. The process $\P$ is simply Glauber dynamics run on the graph $\Z^d[B']$ (the subgraph of $\Z^d$ induced by the set $B'$) with `$+$' boundary conditions.

Next we shall define the process $\Q$ on the block $B'$. Let $A^+$ denote the set of vertices initially in state `$+$' in $B'$, and let $A^-$ denote the set of vertices initially in state `$-$', so $A^- = B' \setminus A^+$. Let $S^{(0)} = A^-$, let $m = \ds\frac{\eps d}{24}$, run the $\textup{Boot}(d,8,m)$ process, defined above, on the graph $G = [n']^d$ (i.e., the \emph{torus} with vertex set $B'$), and let $X = S^{(8)} \setminus A^-$. Finally, let the state $\sigma(x)$ of a vertex $x \in B'$ be declared `$-$' after the process $\Q$ if any of the following is true:
\begin{itemize}
\item Its clock has not yet rung in $\P$ by time $d$.
\item It has at least $d$ neighbours in $A^-$.
\item It has at least one neighbour in $X$.
\end{itemize}
Let $Z$ denote the set of vertices in $B'$ whose state is declared `$-$' after $\Q$.

Let $F$ denote the event that there exists a vertex in $B'$ whose state  is `$-$' at time $d$ in $\P$, but not after the process $\Q$. We shall use the following result, which follows immediately from Lemma 6.3 of~\cite{Maj}.

\begin{lemma}\label{tool}
Let $N,d \in \N$, and let $G = [N]^d$. Let $\eps > 0$ and $p = \ds\frac{1}{2} - \eps$, and choose the elements of $S^{(0)} \subset V(G)$ independently at random, each with probability $p$. Further, let $m = \ds\frac{\eps d}{24}$ and $1 \le k \le 8$. Then, in the $\textup{Boot}(d,8,m)$ process, for every $x \in V(G)$,
$$\Pr\big(x \in S^{(k+1)} \setminus S^{(k)} \big) \; \le \; \exp\left( -\frac{\eps^{k+2} d^{k+1}}{8^{2k+1} (k+1)!} \right).$$
\end{lemma}

From this point onwards, let $\eps > 0$ be arbitrary, let $p = \ds\frac{1}{2} + \eps$, and let the elements of $A^+ \subset B'$ be chosen independently at random, each with probability $p$. We shall denote by $\Pr_p$ probabilities which come from this distribution.

We begin by showing that $\Q$ is almost always more generous than $\P$ (in the trivial coupling). Recall that $F$ denotes the event that there exists a vertex in $B'$ whose state is `$-$' at time $d$ in $\P$, but not after $\Q$.

\begin{lemma}\label{F}
Suppose $\eps^2 d \ge 10^{10}\log d$. Then
$$\Pr_p(F) \; \le \; (2n)^d \exp\left( -\frac{\eps^{10} d^9}{8^{13}\,9!} \right) \; \le \; \exp\big( - d^4 \big).$$
\end{lemma}

\begin{proof}
Let $x \in B'$, and suppose that $\sigma(x)$ is `$-$' after time $d$ in $\P$, but that $\sigma(x)$ is `$+$' after $\Q$. By the definition of $\Q$, the clock of $x$ must have rung at least once before time $d$, and $x$ must have fewer than $d$ neighbours in $A^-$ in the torus on $B'$. Therefore it also had fewer than $d$ neighbours in $A^-$ in the graph $\Z^d[B']$ with `$+$' boundary conditions. But its state after time $d$ in $\P$ is `$-$', so it must have gained a new `$-$' neighbour, $y$ say, in $\P$. Note that $y \notin X$, since $\sigma(x)$ is `$+$' after $\Q$.

Now, since the state of vertex $y$ changed to `$-$' in $\P$, it must lie in the closure of the set $A^-$ under the $d$-neighbour bootstrap process on $\Z^d[B']$. Hence it also lies in the closure of $A^-$ under the $\textup{Boot}(d,8,m)$ process on the torus (since the original process is dominated by the modified one). Let $S^{(0)} = A^-$ and apply the $\textup{Boot}(d,8,m)$ process on the torus. By Lemma~\ref{tool} we have, for each $z \in B'$,
$$\Pr_p(z \in S^{(9)} \setminus S^{(8)}) \; \le \; \exp\left( -\frac{\eps^{10} d^9}{8^{13}\, 9!} \right).$$
Thus, since $|B'| \le (2n)^d$,
$$\Pr_p(|S^{(9)} \setminus S^{(8)}| \ge 1) \; \le \; (2n)^d \exp\left( -\frac{\eps^{10} d^9}{8^{13}\, 9!} \right).$$
But if $S^{(9)} \setminus S^{(8)} = \emptyset$, then all vertices in the closure of $A^-$ (and not in $A^-$) are also in $X$ (by the definition of $X$). But this implies that $y \in X$, which is a contradiction. Thus the event $F$ is contained in the event $S^{(9)} \setminus S^{(8)} \neq \emptyset$, and the result follows.
\end{proof}

Next we show that the set $X = S^{(8)} \setminus S^{(0)}$ is likely to be small.

\begin{lemma}\label{X}
Let $x \in B'$, and suppose $\eps^2 d \ge 10^{10}\log d$. Then
$$\Pr_p(x \in X) \; \le \; 2\exp( -2\eps^2 d) \; < \; \left( \frac{1}{d} \right)^{1000}.$$
\end{lemma}

\begin{proof}
We apply Lemma~\ref{tool} to the torus $[n']^d$ on vertex set $B'$. Recall that the elements of $S^{(0)} = A^-$ are chosen independently at random with probability $1 - p = 1/2 - \eps$. Thus, by Chernoff's inequality,
$$\Pr_p\big(x \in S^{(1)} \setminus S^{(0)} \big) \; \le \; \Pr_p\big(\textup{Bin}(2d,1-p) \ge d\big) \; \le \; \exp\left( - 2 \eps^2 d \right).$$
Thus, by Lemma~\ref{tool},
\begin{eqnarray*}
\Pr_p(x \in X) & \le & \sum_{m = 0}^7 \Pr_p\left(x \in S^{(m+1)} \setminus S^{(m)} \right)\\
& \le &  \exp\left( - 2\eps^2 d \right) + \sum_{m = 1}^7 \exp\left( -\frac{\eps^{m+2} d^{m+1}}{8^{2m+1} (m+1)!} \right)\\
& \le &  2\exp\left( - 2\eps^2 d \right),
\end{eqnarray*}
since $\eps^2 d \ge 10^{10}\log d$, as required.
\end{proof}

Finally we show that $Z$, the set of vertices in $B'$ whose state is `$-$' after $\Q$, is likely to be small.

\begin{lemma}\label{qbound}
Let $x \in B'$, and suppose $\eps^2 d \ge 10^{10}\log d$. Then
$$\Pr_p \big( \sigma(x) \textup{ is `$-$' after }\Q \big) \; \le \; 5d\exp( -2\eps^2 d ) \; < \; \left( \frac{1}{d} \right)^{1000}.$$
\end{lemma}

\begin{proof}
There are three ways in which a vertex can be declared to be in state `$-$' after $\Q$, and each of them is unlikely. Indeed,
\begin{itemize}
\item Since the clocks are exponential, the probability a given clock hasn't yet rung by time $d$ is $e^{-d}$.
\item Since the elements of the set $A^-$ are chosen independently at random with probability $1/2 - \eps$, and each vertex has $2d$ neighbours, the probability a vertex has at least $d$ neighbours in $A^-$ is at most $\exp( -2\eps^2 d )$, by Chernoff's inequality.
\item By Lemma~\ref{X}, the probability that a vertex had a neighbour in $X$ is at most
$$\sum_{y \in \Gamma(x)} \Pr_p(y \in X) \; = \; 2d \Pr_p(y \in X) \; \le \; 4d\exp( -2\eps^2 d).$$
\end{itemize}
The result follows by summing these three probabilities.
\end{proof}

Define $q := \sup_{y \in B'} \Pr_p(y \in Z)$, so we have $q \le d^{-1000}$, by Lemma~\ref{qbound}. We finish the section with a trivial, but crucial observation.

\begin{obs}\label{15ind}
Let $G$ be the torus on vertex set $B'$. The events $\{(x \in Z) : x \in V(G)\}$ are 19-independent.
\end{obs}

\section{From time $d$ to time $O(d^5)$}\label{afterd}

Let $B$ and $B'$ be as described in Section~\ref{firstd}, and let $Y$ denote the set of vertices in $B'$ in state `$-$' after running the process $\P$, i.e., Glauber dynamics on $\Z^d[B']$ with `$+$' boundary conditions, up to time $d$. In the previous section we proved that, if $\eps^2 d \ge 10^{10}\log d$, then there exists a (random) set $Z \subset B'$ which satisfies $\Pr_p(Z \not\supset Y) \le e^{-d^4}$, $\Pr_p(x \in Z) \le d^{-1000}$ for each $x \in B'$, and which is 19-independent. In this section we shall deduce that, after enough extra time, the entire block $B$ will be in state `$+$' with high probability.

We begin by showing that, for each vertex $x \in B'$, the probability that $\sigma(x)$ is `$-$' in the process $\P$ at \emph{any} time $t \ge d$ is small. Again we use a coupling argument in order to retain long-range independence. Let $[Z]_{40}$ denote the closure of the set $Z$ after 40 steps of the $\textup{Boot}(d,40,m)$ process on the torus on $B'$ (i.e., the set $S^{(40)}$ given $S^{(0)} = Z$), where $m = d/80$. 

(We remark that the number 40 is simply chosen to be sufficiently large compared with 19, and sufficiently small compared with $d$. Indeed, in the proof of Lemma~\ref{F'}, below, we shall use the inequality $|T| \ge m^t/2^t t! \ge d^3k$ for $t = 40$, where $m = d/80$ and $k \approx (2d)^{18}$ is the number of points within distance 18 of a vertex in $\Z^d$.)

Let $F'$ denote the event that, in the process $\P$, any vertex outside $[Z]_{40}$ is ever again in state `$-$' after time $d$. We shall need the following simple approximation.

\begin{obs}\label{nunlikely}
Let $p \in (0,1)$ and $n \in \N$ satisfy $pn^2 \le 1$, and let $S(n) \sim \textup{Bin}(n,p)$. Then
$$\Pr_p\big(S(n) \ge m\big) \; \le \; 2p^{m/2}$$ for every $m \in [n]$.
\end{obs}

\begin{proof}
We have 
$$\Pr_p\big( S(n) \ge m \big) \; \le \; \ds\sum_{i=m}^n \ds{n \choose i} p^i \; \le \; 2(pn)^m \,\le\, 2p^{m/2},$$ as claimed. The second inequality follows since $pn \le 1/2$, and the third since $pn \le \sqrt{p}$.
\end{proof}

The following lemma uses ideas from Lemmas 6.2 and 6.3 of~\cite{Maj}.

\begin{lemma}\label{F'}
Suppose $\eps ^2 d \ge 10^{10}\log d$. Then
$$\Pr_p(F') \; \le \; 2\exp\big( - d^4 \big).$$
\end{lemma}

\begin{proof}
We shall prove the lemma using Lemma~\ref{F}, and the following claim.\\[-1ex]

\noindent\ul{Claim}: Let $S^{(0)} = Z$ and $m = d/80$. Then, in the $\textup{Boot}(d,40,m)$ process,
$$\Pr_p\big( | S^{(41)} \setminus S^{(40)} | \ge 1 \big) \; \le \; \exp \big( - d^4 \big).$$

\begin{proof}[Proof of claim]
Recall that $q = \sup_{y \in B'} \Pr_p(y \in Z) \le d^{-1000}$, and suppose that $x \in S^{(41)} \setminus S^{(40)}$. We start by showing that there exists a set $T \subset S^{(1)} \setminus S^{(0)}$, with $d(x,y) = 40$ for each $y \in T$, such that
$$|T| \; \ge \; \frac{m^{40}}{2^{40} 40!} \; \ge \; \frac{d^{40}}{10^{140}}.$$
Indeed, writing $\Gamma(x,j) := \{v \in B' : d(x,v) = j\}$ for each $j \in \N$, let
$$T_j = \Gamma(x,j) \cap S^{(41-j)} \setminus S^{(40-j)},$$ and observe that $|T_1| \ge m$, i.e., that $\Gamma(x)$ must contain at least $m$ vertices of $S^{(40)} \setminus S^{(39)}$. To see this, simply note that if $x \notin S^{(40)}$ then $|\Gamma(x) \cap S^{(39)}|$ is at most $d - m$, and if $x \in S^{(41)} \setminus S^{(40)}$ then $|\Gamma(x) \cap S^{(40)}|$ is at least $d$.

Now, in exactly the same way, for each vertex $y \in T_j$, $\Gamma(y)$ must contain at least $m$ vertices of $S^{(40-j)} \setminus S^{(39-j)}$. At least $m - j \ge m/2$ of these are in $\Gamma(x,j+1)$ (since $y$ has at most $j$ neighbours outside $\Gamma(x,j+1)$), and therefore also in $T_{j+1}$. Since each vertex at distance $j+1$ from $x$ has at most $j+1$ neighbours in $\Gamma(x,j)$, it follows that
$$|T_{j+1}| \; \ge \; \frac{m|T_j|}{2(j+1)}.$$
Thus we obtain the set $T = T_{40}$, as claimed.

Now, consider the set $U = \Gamma(T) \cap \Gamma(x,41)$, and partition $U$ into sets $U_1, \ldots, U_k$, where $k \le 2(2d)^{18}$, so that if $y,z \in U_j$ for some $j$ then $d(y,z) \ge 19$. (That we can do so follows from the simple fact that $\chi(G) \le \Delta(G) + 1$, see for example Lemmas 3.6 and 6.1 of \cite{Maj}.) Since $T \subset S^{(1)} \setminus S^{(0)}$, each vertex of $T$ has at least $d - 40m = d/2$ neighbours in $S^{(0)} = Z$. Also, since $T \subset \Gamma(x,40)$, each vertex of $T$ sends at most 40 edges outside $U$.

It follows that there are at least $(d/2 - 40)|T| \ge d|T|/3$ edges from $T$ to $U \cap Z$. Moreover, each vertex of $U$ sends at most 41 edges into $T$, and so $U$ contains at least $d|T| / 123$ vertices of $Z$. By the pigeonhole principle, for some set $U_j$ we have
$$|U_j \cap Z| \; \ge \; \frac{d|T| }{ 123k } \; \ge \; d^4$$
since $d \ge 10^{10}$.

But the events $\{(y \in Z) : y \in U_j\}$ are independent, by Observation~\ref{15ind}, and
$$|U_j|^2q \; \le \; (2d)^{80} \left( \frac{1}{d} \right)^{1000} \; \le \; 1,$$ so by Observation~\ref{nunlikely},
$$\Pr_p\big( |U_j \cap Z| \ge d^4 \big) \; \le \; 2q^{d^4 / 2} \; \le \; e^{-2d^4}.$$
Now, we have at most $(2n)^d \le e^{d^2}$ choices for the vertex $x$, and at most $m \le 2(2d)^{18}$ choices for the set $U_j$. Thus
$$\Pr_p\big( | S^{(41)} \setminus S^{(40)} | \ge 1 \big) \; \le \; \Big( e^{d^2} 2(2d)^{18} \Big) \Pr_p\big( |U_j \cap Z| \ge d^4 \big) \; \le \; e^{-d^4},$$
as claimed.
\end{proof}

Now, recall that the event $F$ has probability at most $e^{-d^4}$, by Lemma~\ref{F}, and assume that $F$ does not hold, so $Y \subset Z$. Thus the sites ever again in state `$-$' after time $d$ in the process $\P$ are a subset of $[Y] \subset [Z]$, the closure under the usual majority bootstrap rule. But if $S^{(41)} \setminus S^{(40)} = \emptyset$, then $[Z] \subset [Z]_{40}$, and it follows that $F'$ does not hold. Hence
$$\Pr_p(F') \; \le \; \Pr_p(F) \, + \, \Pr_p\big( S^{(41)} \setminus S^{(40)} \neq \emptyset \big) \; \le \; 2\exp \big( - d^4 \big)$$
by Lemma~\ref{F} and the claim, as required.
\end{proof}

We now bound the probability that a vertex is contained in $[Z]_{40}$.

\begin{lemma}\label{z30}
Let $x \in B'$, and suppose $\eps ^2 d \ge 10^{10}\log d$. Then
$$\Pr_p(x \in [Z]_{40}) \; \le \; \left( \frac{1}{d} \right)^{500}.$$ 
and the events $x \in [Z]_{40}$ are $120$-independent.
\end{lemma}

\begin{proof}
If $x \in [Z]_{40}$, then there must exist an element of $Z$ within distance 40 of $x$. But the expected number of such elements is at most $2(2d)^{40}q$, and so
$$\Pr_p\big(x \in [Z]_{40}\big) \; \le \; 2d(2d)^{40}q \; \le \; \left( \frac{1}{d} \right)^{500}.$$
The event $x \in [Z]_{40}$ depends only on vertices within distance $58$ of $x$, so these events are 120-independent.
\end{proof}

Finally, we deduce the bound we require.

\begin{lemma}\label{inact}
Let $x \in B'$, and suppose $\eps^2 d \ge 10^{10}\log d$. Then
$$\Pr_p\big( \sigma(x) \textup{ is `$-$' at time }200d^5 + d\textup{ in }\P \big) \; \le \; 3\exp\big( -d^4 \big).$$
\end{lemma}

\begin{proof}
Let $T = d^5$, and suppose that $\sigma(x)$ is `$-$' at time $200T + d$. Let $E$ denote the event that, at some point before time $200T + d$, a time interval of length $T$ passes in which the clock of some vertex within distance 200 of $x$ does not ring. There are at most $2(2d)^{200}$ such vertices, and if such an interval occurs then it contains an interval of the form $[Tj/2,T(j+1)/2]$. There are $400$ such intervals, and the probability that a given clock does not ring in one of them is $\exp( -T/2 )$. Hence,
$$\Pr_p(E) \; \le \; 800 (2d)^{200} \exp\left( - \frac{d^5}{2} \right) \; \le \; e^{-2d^4}.$$
For the rest of the proof, assume that $E$ does not occur. Assume also that $F'$ does not hold, so if $\sigma(y)$ is `$-$' at some time $t \ge d$, then it follows that $y \in [Z]_{40}$.

Since $E$ does not occur, the clock of $x$ rings at some point in the interval $[199T+d,200T+d)$. Let $t(x)$ denote the last time this happens before $200T+d$, and observe that, since $\sigma(x)$ is `$-$' at time $200T + d$, $x$ must have a set $R(1)$ of at least $d$ neighbours in state `$-$' at time $t(x) \ge 199T + d$. Similarly, each clock associated with a vertex of $R(1)$ rings at some point in the interval $[t(x)-T,t(x))$. For each vertex $y$, let $t(y)$ denote the last time this happens, and observe that at time $t(y)$ vertex $y$ has at least $d$ neighbours in state `$-$', of which at least $d-1 \ge d/2$ are at distance two from $x$ (since it has only one neighbour outside $\Gamma(x,2)$). Each vertex in $\Gamma(x,2)$ has only two neighbours in $\Gamma(x,1)$, and so there is a set $R(2) \subset \Gamma(x,2)$ of at least $|R(1)|d/4$ vertices, which are each in state `$-$' at some time after $198T + d$.

In general, for each $1 \le j \le 199$ and each vertex $z \in R(j) \subset \Gamma(x,j)$, there exists a time $t(z) \ge (200 - j - 1)T + d$ at which the clock of vertex $z$ rings, and $z$ has at least $d$ neighbours in state `$-$', of which at least $d-j \ge d/2$ are at distance $j+1$ from $x$ (since $z$ has only $j$ neighbours outside $\Gamma(x,j+1)$). Each vertex in $\Gamma(x,j+1)$ has at most $j+1$ neighbours in $\Gamma(x,j)$, and so there is a set $R(j+1) \subset \Gamma(x,j+1)$ of at least
$$\frac{|R(j)|d}{2(j+1)}$$
vertices, which are each in state `$-$' at some time after $(200 - j - 1)T + d$.

From this process (see also the proof of Lemma~\ref{F'}), we obtain sets $R(k) \subset \Gamma(x,k)$ for each $k \in [200]$, such that for each vertex $y \in R(k)$, $\sigma(y)$ is `$-$' at some time $t \ge (200 - k)T + d$. Moreover, we have
$$|R(k)| \; \ge \; \frac{d^k}{2^k \,k!}$$ for each $k \in [200]$. Finally, note that each vertex of $R(k)$ is in state `$-$' at some time after $d$, so must also be in $[Z]_{40}$.

Now, let $U = R(200)$, and partition $U$ into sets $U_1, \ldots, U_m$, where $m \le 2(2d)^{119}$, so that if $y,z \in U_j$ for some $j \in [m]$, then $d(y,z) \ge 120$ in the torus on $B'$. Observe that, by the pigeonhole principle, some set $U_j$ contains at least 
$$\frac{|R(200)|}{m} \; \ge \; \left( \frac{d^{200}}{2^{200} 200!} \right) \left( \frac{1}{2(2d)^{119}} \right) \; \ge \; d^{20}$$ vertices of $[Z]_{40}$, since $d \ge 10^{10}$.

But the events $\{(y \in [Z]_{40}) : y \in U_j\}$ are independent, and, by Lemma~\ref{z30},
$$|U_j|^2 \Pr_p(y \in [Z]_{40}) \; \le \; (2d)^{400} \left( \frac{1}{d} \right)^{500} \; \le \; 1$$
for every $y \in B'$. Thus, by Observation~\ref{nunlikely},
$$\Pr_p\big( |U_j \cap [Z]_{40}| \ge d^{20} \big) \; \le \; 2 \Pr_p\big(x \in [Z]_{40}\big)^{d^{20}/2} \; \le \; e^{-3d^4}.$$
Finally, we have at most $m \le 2(2d)^{119}$ choices for the set $U_j$. Thus
\begin{eqnarray*}
\Pr_p\big( \sigma(x) \textup{ is `$-$' at time }200d^5 + d\textup{ in }\P \big) & \le & \Pr_p(E) \, + \, \Pr_p(F') \, + \, 2(2d)^{119} e^{-3d^4} \\
& \le & 3\exp\big( - d^4 \big),
\end{eqnarray*}
by Lemma~\ref{F'}, as required.
\end{proof}

\section{The proof of Theorem~\ref{to1/2}}\label{proofsec}

In this section we shall put together the pieces and prove Theorem~\ref{to1/2}. We have shown that, in the process $\P$, for any vertex $x \in B'$,
$$\Pr_p\big( \sigma(x)\textup{ is `$-$' at time }200d^5 + d \big) \; \le \; 3 \exp\big( - d^4 \big).$$
Thus the probability that there exists a vertex in $B'$ in state `$-$' at this time is at most $\exp( - d^4/2 )$, since $B'$ has $(n')^d \le e^{d^2}$ vertices. However, this is in the process $\P$, not the original Glauber dynamics. We therefore need one more lemma. (See also Step 3 of the proof of Lemma 4.1 in~\cite{FSS}, on which the following lemma is based.)

Define a \emph{path of clock-rings} to be a sequence $(x_1,t_1), \ldots, (x_m,t_m)$ of vertex-time pairs, where $x_j \in \Z^d$ and $t_j \in [0,\infty)$, such that the following conditions hold:
\begin{itemize}
\item $\|x_{j+1} - x_j\|_1 = 1$ for each $j \in [m-1]$.\\[-2ex]
\item $t_1 < \dots < t_m$.\\[-2ex]
\item The clock of vertex $x_j$ rings at time $t_j$ for each $j \in [m]$.
\end{itemize}
We say moreover that such a sequence is a path from $x_1$ to $x_m$ in time $[t_1,t_m]$. We begin with a simple but key observation.

\begin{obs}\label{path}
Let $x,y \in \Z^d$ and $t \in [0,\infty)$. Suppose that there does not exist a path of clock-rings from $x$ to $y$ in time $[0,t]$. Then the state of vertex $y$ at time $t$ is independent of the state of vertex $x$ at time $0$.
\end{obs}

Let $F''$ denote the event that there exists a path of clock-rings from some vertex outside $B'$ to some vertex inside $B$ in time $[0,T]$, where $T = 200d^5 + d$. Note that, by Observation~\ref{path}, if $F''$ does not occur, then the state of every vertex in $B$ at time $T$ is the same in Glauber dynamics on $\Z^d$ as it is in the process $\P$, since the boundary conditions cannot affect $B$.

\begin{lemma}\label{F''}
$\Pr_p(F'') \, \le \, 2^{-2^d}.$
\end{lemma}

\begin{proof}
For each $r \in \N$, there are at most $(2n)^d(2d)^r$ paths of length $r$ starting on the boundary of $B'$. Given a time $T \in [0,\infty)$, let $P(r,T)$ denote the probability that a particular path of length $r$, $(x_1, \ldots, x_r)$ say, can be extended to a path of clock-rings in time $[0,T]$. In other words, $P(r,T)$ is the probability that there exist times $0 \le t_1 < \dots < t_r \le T$ such that $(x_1,t_1), \ldots, (x_r,t_r)$ is a path of clock-rings. It is clear that $P(r,T)$ does not depend on the particular path we choose.

We bound $P(r,T)$ as follows. For each $j \in [r]$ choose $t_j$ to be the first time the clock $C(x_j)$ rings after time $t_{j-1}$. Let $J_k$ denote the event that $t_k - t_{k-1} \le 2T/r$, and observe that
$$\Pr_p(J_k) \; = \; 1 \, - \, \exp\left( -\frac{2T}{r} \right) \; \le \; \frac{2T}{r},$$
and that the events $J_k$ are independent. Let $J = \sum_{k=1}^r I[J_k]$, where $I$ denotes the indicator function. Then,
$$P(r,T) \; = \; \Pr_p\big( t_r \le T \big) \; \le \; \Pr_p\left(J \ge \frac{r}{2} \right) \; \le \; {r \choose r/2} \left( \frac{2T}{r} \right)^{r/2} \; \le \; \left( \frac{8T}{r} \right)^{r/2}.$$

Now, applying this with $r \ge 2^d$ and $T = 200d^5 + d$, we obtain
$$\Pr_p(F'') \; \le \; \sum_{r = 2^d}^\infty (2n)^d(2d)^r \left( \frac{8T}{r} \right)^{r/2} \; \le \; 2^{-2^d}$$ as required.
\end{proof}

Finally, we are ready to prove Theorem~\ref{to1/2}.

\begin{proof}[Proof of Theorem~\ref{to1/2}]
Let $\eps > 0$ and let $p = \ds\frac{1}{2} + \eps$. Let $d \in \N$ satisfy $\eps^2 d \ge 10^{10}\log d$, and choose the elements of the set $A^+ \subset \Z^d$ independently at random, each with probability $p$. Let $n = 3 \times 2^d$, and partition $\Z^d$ into blocks of size $[n]^d$, in the obvious way.

We run Glauber dynamics for time $T = 200d^5 + d$, and then stop. Given a block $B$, define the block $B' \supset B$, and the process $\P$ on $B'$, as in Section~\ref{firstd}. We say that $B$ is a good block if both of the following events occur in $B'$:
\begin{itemize}
\item The event $F''$ does not occur.\\[-2ex]
\item All of the elements of $B$ are in state `$+$' at time $T$ in the process $\P$.
\end{itemize}
Otherwise we say that $B$ is a bad block.

Note that if $B$ is good, then all the elements of $B$ are in state `$+$' at time $T$ in Glauber dynamics, by the comment after Observation~\ref{path}. Also, by Lemmas~\ref{inact} and \ref{F''}, the probability that $B$ is bad is at most
$$\Pr_p(F'') \, + \, \sum_{x \in B} \Pr_p\big( \sigma(x) \textup{ is `$-$' at time } T \textup{ in }\P \big) \; \le \; 2^{-2^d} \, + \, 3n^d \exp\big( -d^4 \big) \; \le \; \exp \left( - \frac{d^4}{2} \right).$$
Moreover, the event ``$B$ is good" depends only on what happens inside $B'$. Hence, given any collection of blocks $B_1,\ldots,B_k$ with $\|B_i - B_j\|_\infty \ge 2$ for each $i \neq j$, the events ``$B_j$ is good" are independent, since the corresponding blocks $B_j'$ are all disjoint.

Hence we may couple the dynamics at time $T$ with a distribution $\sigma \in \Omega(n,p)$, where $p = \Pr_p(B\textup{ is good})$. But
$$\Pr_p\big( B \textup{ is good}\big) \; \ge \; 1 \,-\, \exp\left( - \frac{d^4}{2} \right) \; > \; 1 \, - \, \eps' \left( \frac{1}{2n} \right)^{d^2} \; \ge \;  p^{(n)}_c(\Z^d),$$
by Theorem~\ref{FSSthm}, and so the system fixates at `$+$' with probability $1$, as required.
\end{proof}

\section{Acknowledgements}

The author would like to thank Vladas Sidoravicius for suggesting the problem to him, for reading an early version of the manuscript, and for several stimulating discussions. He would also like to thank the anonymous referee for a very careful reading of the proof, and for many useful comments.

\end{document}